\documentclass[conference]{IEEEtran}
\usepackage{cite}
\usepackage{graphicx}
\usepackage{psfrag}
\usepackage{url}
\usepackage{amsmath}
\usepackage{array}
\usepackage{amssymb}
\usepackage{amsfonts}
\usepackage{epstopdf}

\newtheorem{lemma}{Lemma}

\newtheorem{remark}{Remark}
\newtheorem{definition}{Definition}
\newtheorem{theorem}{Theorem}

\usepackage{subfigure}

\title{Interference Alignment with Diversity for the $2 \times 2$ $X$-Network with three antennas }

\begin{document}

\author{
\authorblockN{Abhinav Ganesan and B. Sundar Rajan\\}
\IEEEauthorblockA{\small{Email: \{abhig\_88, bsrajan\}@ece.iisc.ernet.in}}
}

\maketitle
\thispagestyle{empty}	

\begin{abstract}

Interference alignment is known to achieve the maximum sum DoF of $\frac{4M}{3}$ in the $2 \times 2$ $X$-Network (i.e., two-transmitter (Tx) two-receiver (Rx) $X$-Network) with $M$ antennas at each node, as demonstrated by Jafar and Shamai. Recently, an Alamouti code based transmission scheme, which we call the Li-Jafarkhani-Jafar (LJJ) scheme, was proposed for the $2 \times 2$ $X$-Network with two antennas at each node. This scheme achieves a sum degrees of freedom (DoF) of $\frac{8}{3}$ and also a diversity gain of two when fixed finite constellations are employed at each Tx. In the LJJ scheme, each Tx required the knowledge of only its own channel unlike the Jafar-Shamai scheme which required global CSIT to achieve the maximum possible sum DoF of $\frac{8}{3}$. Bit error rate (BER) is an important performance metric when the coding length is finite.
This work first proposes a new  STBC for a three transmit antenna single user MIMO system. Building on this STBC, we extend the LJJ scheme to the $2 \times 2$ $X$-Network with three antennas at each node. Local channel knowledge is assumed at each Tx. It is shown that the proposed scheme achieves the maximum possible sum DoF of $4$. A diversity gain of $3$ is also guaranteed when fixed finite constellation inputs are used. 
\end{abstract}	

\section{Introduction}  \label{sec1}
Maximizing the rate\footnote{The definition of rate used here is in the sense of vanishing probability of error as given in Section $7.5$ of \cite{CoT-book}.} of transmission and minimizing the bit error rate have been intensely pursued in single user communication systems. Multiple-input multiple-output (MIMO) systems offer the potential to improve both and hence, are immensely popular. Bit error rate is an issue in the case of coding over a fixed number of time slots while achievable rate is a performance metric when coding over infinite time slots is allowed. So, the design of space-time block codes (STBC) in single user MIMO systems incorporated two properties as performance metrics - the information losslessness property of an STBC and the diversity gain offered by an STBC \cite{DTB, SRS}. The former assures that the maximum achievable rate is not sacrificed by introduction of the STBC block in the MIMO system while the latter assures a degree of reliability in the data when the code length is 
restricted 
and SNR-independent finite constellation input is used. 

Many recent works on multiuser communication systems, specifically on interference networks, have focused on sum-capacity optimal or approximate sum-capacity optimal transmission strategies. The notion of approximate sum-capacity is captured by the concept of degrees of freedom (DoF). The sum DoF of a Gaussian interference network is said to be $d$ if the sum-capacity can be written as $d~log_2 SNR + o(log_2 SNR)$ \cite{JaS}. Wireless $X$-Networks are a class of Gaussian interference networks with $K$ transmitters and $J$ receivers where every receiver demands an independent message from every transmitter so that there is a total of $KJ$ messages meant to be transmitted in the network. We shall denote an $X$ network with $M$ antennas at every node by $(K,J,M)-X$-Network. A sum DoF of $\lfloor{\frac{4M}{3}}\rfloor$ was shown to be achievable in a $(2,2,M)-X$-Network in \cite{MMK} while Jafar and Shamai in \cite{JaS} showed that a sum DoF of $\frac{4M}{3}$ is achievable using the idea of interference alignment 
(IA)\footnote{We shall henceforth call the transmission scheme proposed in \cite{JaS} as the JS scheme.}. Further, $\frac{4M}{3}$ is proven \cite{JaS} to be an outerbound on the sum DoF of $(2,2,M)-X$-Network which establishes $\frac{4M}{3}$ to be the sum DoF of $(2,2,M)-X$-Network. 

Recently, Alamouti codes were intelligently coupled with channel dependent precoding \cite{LJJ,LiJ} to achieve a sum-DoF of $\frac{8}{3}$ in the $(2,2,2)-X$-Network. Also, a diversity gain of two was assured. The transmission scheme in \cite{LJJ} shall be referred to as the LJJ scheme. In a sense, the performance metrics that the LJJ scheme highlights i.e., achieving the maximum sum DoF and a diversity gain that is strictly greater than one, is akin to the properties sought after in the design of STBCs for single user communication systems i.e., the information losslessness property and the diversity gain. Note that sum DoF is an approximate sum-capacity at high SNR. Thus, STBCs with information losslessness property that were sought after in single user systems is now translated to approximate sum-capacity lossless design of STBCs in the  $(2,2,2)-X$-Network. The design of STBCs that offer diversity gain in $(2,2,M)-X$-Network however comes with a crucial difference with respect to the single user scenario. 
Local 
channel knowledge (i.e., every transmitter knows only its own channel) is assumed in the LJJ scheme for the $(2,2,2)-X$-Network whereas the design of information lossless STBCs in the single user set-up does not assume CSIT. Similarly, an extension of the LJJ scheme for the $(2,2,4)-X$-Network \cite{AbR} using the Srinath-Rajan (SR) STBC \cite{SrR1} assumed local channel knowledge. Assumption of no CSIT would make things difficult. Even the JS scheme assumed global CSIT, i.e., knowledge of all the channel gains at all the transmitters, to achieve sum DoF of $\frac{4M}{3}$.

The challenge in extending the LJJ scheme to a general $(2,2,M)-X$-Network is to identify STBCs that could be applied in the $(2,2,M)-X$-Network with appropriate modifications. This was done for $M=4$ in \cite{AbR} where the SR STBC fitted nicely into the extended LJJ scheme. This work aims to extend the LJJ scheme to the $(2,2,3)-X$-Network. As in the LJJ scheme and the extended LJJ scheme for the $(2,2,4)-X$-Network using the SR STBC, the performance metrics are the achievable sum DoF and the diversity gain. The contributions of this work are summarized below. 

\begin{itemize}
 \item We propose an STBC that encodes $\frac{3}{2}$ complex symbols per channel use (cspcu) for a three transmit antenna single user MIMO system. We use this STBC in the extended LJJ scheme to achieve the maximum sum DoF of $4$ in the $(2,2,3)-X$-Network (Theorem \ref{thm2}, Section \ref{sec4}). Like the LJJ scheme, the proposed scheme assumes only local channel knowledge at the transmitters.
 \item We show that the proposed scheme guarantees a diversity gain of $3$ when fixed finite constellation inputs are used (Theorem \ref{thm1}, Section \ref{sec4}). Simulation results show that the diversity gain is strictly greater than $3$.
\end{itemize}

The paper is organized as follows. The next section formally introduces the system model. Section \ref{sec3} summarizes the JS scheme and the LJJ scheme. The proposed scheme is explained in Section \ref{sec4} where the DoF achievability and the diversity gain achieved are also proved. Section \ref{sec5} presents simulation results illustrating the performance of the proposed scheme. The paper concludes with Section \ref{sec6}.	

{\em Notations:}   The set of complex numbers is denoted by $\mathbb C$. The notation ${\cal CN}(0,\sigma^2)$ denotes the circularly symmetric complex Gaussian distribution with mean zero and variance $\sigma^2$. For a complex number $x$, the notation $\overline{x}$ denotes the conjugate of $x$. The real and imaginary parts of a complex number $a$ are denoted by $a^R$ and $a^I$ respectively. The identity matrix of size $n \times n$ is denoted by $I_n$. An all-zeros column vector is denoted by $\underline{0}$. The trace of a matrix $A$ is denoted by $\text{tr}(A)$. For an invertible matrix $A$, the notation $A^{-H}$ denotes the Hermitian of the matrix $A^{-1}$. The $i^\text{th}$ row and the $i^\text{th}$ column of a matrix $A$ are denoted by $A(i,:)$ and $A(:,i)$ respectively. The $i^\text{th}$ row, $j^\text{th}$ column element of a matrix $A$ is denoted by $A(i,j)$ unless mentioned 
otherwise. The Frobenius norm of a matrix $A$ is denoted by $||A||$. The Kronecker product of two matrices $A$ and $B$ is denoted by $A \otimes B$. A diagonal matrix with the diagonal entries given by $a_1, a_2, \cdots,a_n$ is denoted by $\text{diag}(a_1,a_2,\cdots,a_n)$. 

\section{System Model} \label{sec2}

The $(2,2,M)-X$-Network is shown in Fig. \ref{fig_sys_model}. The message transmitted by transmitter (Tx) $i$ to receiver (Rx) $j$ is represented by $W_{ij}$. 
 \begin{figure}[htbp]
\centering
\includegraphics[totalheight=3.1in,width=2.8in]{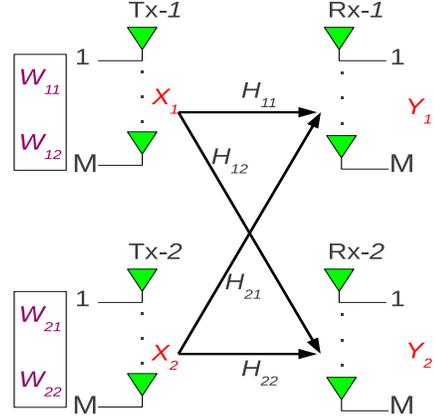}
\vspace{-1cm}
\caption{System Model.}
\label{fig_sys_model}
\end{figure} The input-output relation over $T$ time slots is given by
\begin{align}
\label{eqn-sys_model}
 Y_j=\sqrt{P}\sum_{i=1}^{2} H_{ij}X_i+N_j
\end{align}where, $Y_j \in \mathbb{C}^{M \times T}$ denotes the output symbol matrix at Rx-$j$, $X_i \in \mathbb{C}^{M \times T}$ denotes the input symbol matrix at Tx-$i$ such that $\mathbb{E}\left[\text{tr}\left(XX^H\right)\right]\leq T$, $H_{ij} \in \mathbb{C}^{M \times M}$ denotes the channel matrix between Tx-$i$ and Rx-$j$, $N_j \in \mathbb{C}^{M \times T}$ denotes the noise matrix whose entries are distributed as i.i.d. ${\cal CN}(0,1)$. The average power constraint at each of the transmitters is denoted by $P$. The channel gains are assumed to be constant for the length of the codeword, and the real and imaginary parts of the channel gains are distributed independently according to an arbitrary continuous distribution. Global CSIR is assumed throughout the paper, i.e., both the receivers have the knowledge of all the channel gains. 

\section{Review of some known transmission schemes in $(2,2,M)-X$-Network} \label{sec3}
In the first sub-section, the JS scheme is reviewed and in the second sub-section, the LJJ scheme for the $(2,2,2)-X$-Network is reviewed.

\subsection{Review of JS Scheme for $(2,2,3)-X$-Network} \label{subsec1}
 The JS scheme for $(2,2,3)-X$-Network aligns the interference symbols by precoding over a $3$-symbol extension of the channel, i.e., $T=3$. Every transmitter sends $3$ complex symbols to each receiver over $3$ channel uses so that a sum DoF of $4$ is attained. The input and the output symbols over a $3$-symbol extension of the channel are related by
\begin{align} \label{eqn-JS}
 Y'_j= \sqrt{\frac{3P}{2}}\sum_{i=1}^{2}H'_{ij} \left(\sum_{k=1}^{2}\frac{V_{ik}}{\text{tr}\left(V_{ik}V^H_{ik}\right)} X_{ik}\right) + N'_j
\end{align}where, $Y'_j \in \mathbb C^{9 \times 1}$ denotes the received symbol vector at Rx-$j$ over $3$ channel uses, {\small$H'_{ij}=\begin{bmatrix}
                                                                                                                            H_{ij} & \mathbf{0} & \mathbf{0}\\
															    \mathbf{0} & H_{ij} &\mathbf{0}\\
															    \mathbf{0} & \mathbf{0} & H_{ij}
                                                                                                                           \end{bmatrix}$}
denotes the effective channel matrix between Tx-$i$ and Rx-$j$ over $3$ channel uses, $V_{ik} \in \mathbb{C}^{9 \times 3}$ denotes a precoding matrix, $X_{ik}\in \mathbb C^{3 \times 1}$ denotes the symbol vector generated by Tx-$i$ meant for Rx-$k$, and $N'_j \in \mathbb{C}^{9 \times 1}$ denotes the Gaussian noise vector whose entries are distributed as i.i.d. ${\cal CN} (0,1)$. The entries of $X_{ik}$ take values from a set such that {\small$\mathbb{E}\left[X_{ik}X^H_{ik}\right]=I_3$}. The precoders $V_{ik}$ are chosen as given below.
{
\begin{align*}
& V_{11}=E^{F'} {V^{F'}_1}, ~~ V_{12}=E^{F'} {V^{F'}_2}, \\
&V_{21}={H'}^{-1}_{22}H'_{12}V_{11}, ~ ~~V_{22} = {H'}^{-1}_{21} H'_{11} V_{12}
\end{align*}}where, {\small$E^{F'} \in \mathbb{C}^{9 \times 9}$} denotes a matrix whose columns are the eigen vectors of the matrix {\small$F'={H'}^{-1}_{11}H'_{21}{H'}^{-1}_{22}H'_{12}$}, \mbox{{\small$V^{F'}_1=I_3 \otimes[1 ~1 ~0]^T$}}, and {\small$V^{F'}_2=I_3 \otimes[1 ~0 ~1]^T$}. Using the above choice of precoders, the interference symbols are aligned and (\ref{eqn-JS}) can be re-written as
\begin{align}
\nonumber
& Y'_1=\sqrt{\frac{3P}{2}}\left(H'_{11}V_{11}X_{11}+H'_{21}V_{21}X_{21}\right.\\
\nonumber
&\hspace{4cm}\left.+H_{11}V_{12}\left(X_{12}+X_{22}\right)\right)+N'_1\\
\label{eqn-JS_Aligned}
& Y'_2=\sqrt{\frac{3P}{2}}\left(H'_{12}V_{12}X_{12}+H'_{22}V_{22}X_{22}\right.\\
\nonumber
&\hspace{4cm}\left.+H_{12}V_{11}\left(X_{11}+X_{21}\right)\right)+N'_2.
\end{align}
The above described scheme is proved to achieve a sum DoF of $4$ in the $(2,2,3)-X$-Network almost surely when the channel matrix entries take values from a continuous probability distribution \cite{JaS}.

\subsection{Review of LJJ Scheme} \label{subsec2}
The LJJ scheme for the $(2,2,2)-X$-Network transmits two superposed Alamouti codes with appropriate precoding in three time slots (i.e., $T=3$) at every transmitter. The symbols meant for each receiver are transmitted through an Alamouti code as given by

{\small
\begin{align*}
X_1&=\sqrt{\frac{3P}{4}}\left(V_{11}\underbrace{\begin{bmatrix}
                          x^1_{11} & -\overline{x^{2}_{11}} & 0\\
			  x^2_{11} &  ~\overline{x^{1}_{11}} & 0
                         \end{bmatrix}}_{X_{11}} + V_{12}\underbrace{\begin{bmatrix}
                          0 & x^1_{12} & -\overline{x^{2}_{12}}\\
			  0 & x^2_{12} &  ~\overline{x^{1}_{12}} 
                         \end{bmatrix}}_{X_{12}} \right)\\
X_2&=\sqrt{\frac{3P}{4}}\left(V_{22}\underbrace{\begin{bmatrix}
                          x^1_{21} & -\overline{x^{2}_{21}} & 0\\
			  x^2_{21} &  ~\overline{x^{1}_{21}} & 0
                         \end{bmatrix}}_{X_{21}} + V_{12}\underbrace{\begin{bmatrix}
                          0 & x^1_{22} & -\overline{x^{2}_{22}}\\
			  0 & x^2_{22} &  ~\overline{x^{1}_{22}} 
                         \end{bmatrix}}_{X_{22}} \right)
\end{align*}}where, $x^k_{ij}$ takes values from a set such that {\small$\mathbb{E}\left[\left|x^k_{ij}\right|^2\right]=1$}. The matrices $X_{ij}$, as defined above, correspond to the symbols generated by Tx-$i$ meant for Rx-$j$. The matrix entries $x^{k}_{ij}$ denote the $k^{\text{th}}$ symbol generated by Tx-$i$ for Rx-$j$. The precoders $V_{ij}$ are chosen as 

{\small
\begin{align}
\nonumber
&V_{11}=\frac{H^{-1}_{12}}{\sqrt{\text{tr}\left(H^{-1}_{12}H^{-H}_{12}\right)}}, ~ V_{12}=\frac{H^{-1}_{11}}{\sqrt{\text{tr}\left(H^{-1}_{11}H^{-H}_{11}\right)}}\\
\label{eqn-precoders}
&V_{21}=\frac{H^{-1}_{22}}{\sqrt{\text{tr}\left(H^{-1}_{22}H^{-H}_{22}\right)}}, ~ V_{22}=\frac{H^{-1}_{21}}{\sqrt{\text{tr}\left(H^{-1}_{21}H^{-H}_{21}\right)}}.
\end{align}}The terms inside the square-roots above ensure that the transmitters meet the average power constraint. We briefly describe how the above choice of precoders and the use of Alamouti codes whose columns are juxtaposed with a zero column help to align the interference.
The output symbol matrix at Rx-$1$ is now given by
{
\begin{align*}
&Y_1=\sqrt{\frac{3P}{4}}H_{11}V_{11}X_{11} + \sqrt{\frac{3P}{4}}H_{21}V_{21}X_{21}  \\
&~~~~~~~~+\sqrt{\frac{3P}{4}}\begin{bmatrix}
                          0 & ax^1_{12} + bx^1_{22} & -a\overline{x^{2}_{12}}-b\overline{x^{2}_{22}}\\
			  0 & ax^2_{12} + bx^2_{22} &  ~a\overline{x^{1}_{12}}+b\overline{x^{1}_{22}} 
                         \end{bmatrix} + N_1
\end{align*}}where, $a=\frac{1}{\sqrt{\text{tr}\left(H^{-1}_{11}H^{-H}_{11}\right)}}$ and $b=\frac{1}{\sqrt{\text{tr}\left(H^{-1}_{21}H^{-H}_{21}\right)}}$. Let the effective channel matrices corresponding to the desired symbols from Tx-$1$ and Tx-$2$ to Rx-$1$ be denoted by $\hat{H}=H_{11}V_{11}$ and $\hat{G}=H_{21}V_{21}$ respectively. Define a matrix $Y' \in \mathbb C^{2 \times 3}$ whose first, second and third columns are given by

{\small
\begin{align}
Y'(:,1)=Y(:,1), ~Y'(:,2)=\overline{Y(:,1)}, ~Y'(:,3)=Y(:,3).
\end{align}}Similarly, define the matrix $N'_1$ obtained from $N_1$. The processed output symbols at Rx-$1$ (i.e., $Y'_1$) can be written as

{\footnotesize
\begin{align}
\nonumber
\underbrace{\begin{bmatrix}
{Y'}_1^T(1,:)\\
{Y'}_1^T(2,:)
\end{bmatrix}}_{Y''_1}=\sqrt{\frac{3P}{4}}&\begin{bmatrix}
\hat{h}_{11} & \hat{h}_{12} & \hat{g}_{11} & \hat{g}_{12} & 0 & 0\\
\overline{\hat{h}_{12}} & -\overline{\hat{h}_{11}} & \overline{\hat{g}_{12}} & -\overline{\hat{h}_{11}} & 1 & 0\\
0 & 0 & 0 & 0 & 0 & -1\\
\hat{h}_{21} & \hat{h}_{22} & \hat{g}_{21} & \hat{g}_{22} & 0 & 0\\
\overline{\hat{h}_{22}} & -\overline{\hat{h}_{21}} & \overline{\hat{g}_{22}} & -\overline{\hat{h}_{21}} & 0 & 1\\
0 & 0 & 0 & 0 & 1 & 0
\end{bmatrix}\begin{bmatrix}
x^{1}_{11}\\
x^{2}_{11}\\
x^{1}_{21}\\
x^{2}_{21}\\
I_1\\
I_2
\end{bmatrix}\\
\label{eqn-Y''_1}
&\hspace{1cm}+\underbrace{\begin{bmatrix}
{N'}_1^T(1,:)\\
{N'}_1^T(2,:)
\end{bmatrix}}_{N''_1}
\end{align}}where, $I_1=a\overline{x^1_{12}}+b\overline{x^1_{22}}$ and $I_2=a\overline{x^2_{12}}+b\overline{x^2_{22}}$, and $\hat{h}_{ij}$ and $\hat{g}_{ij}$ denote the entries of the matrices $\hat{H}$ and $\hat{G}$ respectively. Note that, when $\hat{h}_{ij}$ and $\hat{g}_{ij}$ are non-zero, the interference symbols $I_1$ and $I_2$ are aligned in a subspace linearly independent of the signal subspace. So, pre-multiplying the matrix $Y''_1$ (defined in (\ref{eqn-Y''_1})) by the zero-forcing matrix given by
\begin{align}
F=\begin{bmatrix}
1 & 0 & 0 & 0 & 0 & 0\\
0 & 1 & 0 & 0 & 0 & -1\\
0 & 0 & 1 & 0 & 1 & 0\\
0 & 0 & 0 & 1 & 0 & 0
\end{bmatrix}
\end{align} yields

{\small
\begin{align} \label{eqn-eff_ch_mat}
 FY''_1=\sqrt{\frac{3P}{4}}\underbrace{\begin{bmatrix}
\hat{h}_{11} & \hat{h}_{12} & \hat{g}_{11} & \hat{g}_{12} \\
\overline{\hat{h}_{12}} & -\overline{\hat{h}_{11}} & \overline{\hat{g}_{12}} & -\overline{\hat{h}_{11}} \\
\hat{h}_{21} & \hat{h}_{22} & \hat{g}_{21} & \hat{g}_{22} \\
\overline{\hat{h}_{22}} & -\overline{\hat{h}_{21}} & \overline{\hat{g}_{22}} & -\overline{\hat{h}_{21}} \\
\end{bmatrix}}_{R}\begin{bmatrix}
x^{1}_{11}\\
x^{2}_{11}\\
x^{1}_{21}\\
x^{2}_{21}\\
\end{bmatrix} +FN''_1.  
\end{align}}It is shown that the matrix $R$ is almost surely full rank and hence, a sum DoF of $\frac{8}{3}$ is achieved in the $(2,2,2)-X$-Network. When the symbols $x^k_{ij}$ take values from finite constellations, it is proved that a diversity gain of $2$ can be achieved for every $x^k_{ij}$, along with symbol-by-symbol decoding \cite{LiJ}.

The LJJ scheme is extended to the $(2,2,3)-X$-Network in the next section.

\section{Extended LJJ scheme for the $(2,2,3)-X$-Network} \label{sec4}
In this section, we first present a set of conditions on STBCs, which when coupled with LJJ-type precoding, could be used in a general $(2,2,M)-X$-Network to achieve a diversity gain of $M$. We then present an explicit construction of an STBC for the $(2,2,3)-X$-Network which achieves a diversity gain of $3$ and a sum DoF of $4$ with Gaussian input alphabets. Throughout this section, we shall focus only on decoding the desired symbols and interference cancellation at Rx-$1$. Similar signal processing is assumed to happen at Rx-$2$.

Consider a $M \times T'$ linear dispersion STBC denoted by $X'$ where $T'\geq M$ is an even integer. Linear dispersion STBCs are those which can be represented as $X'=\sum_{i=1}^{L}A^{iR}x^{iR}+A^{iI}x^{iI}$ where, $x^i \in {\cal S}$ represent the symbols that the STBC encodes for some finite constellation ${\cal S}$, and the matrices $A^{iR},A^{iI} \in \mathbb C^{M \times T'}$ are called dispersion matrices \cite{HaH}. Every transmitter uses an STBC with the same structure as that of $X'$ to transmit its message symbols to each of the receivers. Let the STBC corresponding to the symbols meant to be sent from Tx-$i$ to Rx-$j$ be denoted by $X'_{ij}$, $i,j=1,2$. Consider a matrix $X_{i1}$ formed by inserting a zero column to every third column of $X'_{i1}$ so that there are a total of $\frac{T'}{2}$ all-zero columns, i.e., a column of zeros is inserted after every two columns of $X'_{i1}$, starting from its third column. Also, consider a matrix $X_{i2}$ formed by inserting a zero column to every third column 
of $X'_{i2}$, but starting from its first column (i.e., the first column of $X_{i2}$ is an all-zero column), so that there are a total of $\frac{T'}{2}$ all-zero columns. Note that the matrices $X_{ij}$ are $M \times \frac{3T'}{2}$ matrices. The precoders used are the same as in the LJJ scheme, i.e., given by (\ref{eqn-precoders}). The received symbols at Rx-$1$ are given by

{\small
\begin{align*}
 Y_1&=\sqrt{c_1P}H_{11}V_{11}X_{11}+\sqrt{c_2P}H_{21}V_{21}X_{21}\\
 &+\sqrt{c_1P}aX_{12}+\sqrt{c_2P}bX_{22} + N_1
\end{align*}}where, $c_i$ is a normalizing constant which ensures that the power constraint at Tx-$i$ is met. Note that the received symbols at time instants $t=1,4,7,\cdots,\frac{3T'}{2}-2$ are interference-free because of the pattern of the zero columns in $X_{i2}$. The unintended symbols interfere at time instants $t=2,5,8,\cdots,\frac{3T'}{2}-1$, at Rx-$1$. We now narrow down on a desirable structure of the STBC $X'$ so that the interfering symbols can be canceled. Let $f_k:\mathbb{C}\rightarrow \mathbb{C}$ be a deterministic function such that $f_k(w)$ is distributed as ${\cal CN}(0,\sigma_k^2)$ when $w$ is distributed as ${\cal CN}(0,1)$, for $k=1,2,\cdots,\frac{MT'}{2}$. Let $\pi_p$ be a permutation of $\{1,2,\cdots,M\}$, for $p=1,3,\cdots, T'-1$. Suppose there exist functions $f_k$ such that 
\begin{align*}
 X'(:,p)+\begin{bmatrix}
	     f_{\frac{(p-1)M}{2}+1}(X'(\pi_p(1),p+1))\\
	     f_{\frac{(p-1)M}{2}+2}(X'(\pi_p(2),p+1))\\
	     \vdots\\
	     f_{\frac{(p-1)M}{2}+M}(X'(\pi_p(M),p+1))\\
         \end{bmatrix}=\underline{0}
\end{align*} for $p=1,3,\cdots, T'-1$. We call the above property of the STBC $X'$ as the column cancellation property.

We observe that the signal corresponding to the desired symbols is equal to zero at time instants $t=3,6,9,\cdots,\frac{3T'}{2}$, at Rx-$1$ because of the pattern of zero-columns in $X_{i1}$, for $i=1,2$. Thus, on the account of the column cancellation property of $X'$ clearly, the interference symbols at time instants $t=2,5,8,\cdots,\frac{3T'}{2}-1$ can be canceled using the interference symbols received at time instants $t=3,6,9,\cdots,\frac{3T'}{2}$, without affecting the desired symbols at Rx-$1$. Now, the relevant components of the noise vectors corresponding to the time instants $t=2,5,8,\cdots,\frac{3T'}{2}-1$, are distributed as i.i.d.  ${\cal CN}(0,1+\sigma_k^2)$, for $k=1,2,\cdots \frac{MT'}{2}$. Discarding the received symbols at time instants $t=3,6,9,\cdots,\frac{3T'}{2}$, the processed received symbols at Rx-$1$ which are now interference-free can be written as

{\small
\begin{align}
\label{eqn-after_int_cal}
Y'_1=\sqrt{c_1P}H_{11}V_{11}X'_{11}+\sqrt{c_2P}H_{21}V_{21}X'_{21}+N'_1
\end{align}}where, $X'_{ij}$ is obtained from $X_{ij}$ by dropping the all-zero columns. The following theorem states the condition on $X'$ under which ML decoding of $X'_{i1}$ from (\ref{eqn-after_int_cal}) yields a diversity gain of $M$.

Define the difference matrix $\triangle{X'^{k_1,k_2}_{ij}}$ by
\begin{align*}
 \triangle {X'^{k_1,k_2}_{ij}} = {X'^{k_1}_{ij}}-{X'^{k_2}_{ij}}
\end{align*}where, ${X'^{k_1}_{ij}}$ and ${X'^{k_2}_{ij}}$ denote two different realizations (i.e., $k_1 \neq k_2$) of the matrix $X'_{ij}$. We note that the finite constellations involved with the STBCs $X'_{ij}$ for different $(i,j)$ could be different.

\begin{theorem} \label{thm1}
If the channel matrix entries are distributed as  i.i.d. ${\cal CN}(0,1)$ then, the average pair-wise error probability $P_e$ for the distinct pairs of codewords {\small$\left({X'^{k_1}_{11}},{X'^{k_2}_{21}}\right)$} and {\small$\left({X'^{k'_1}_{11}},{X'^{k'_2}_{21}}\right)$} is upper bounded as

{\small\begin{align*}
 P_e\left({\left({X'^{k_1}_{11}},{X'^{k_2}_{21}}\right)\rightarrow\left({X'^{k'_1}_{11}},{X'^{k'_2}_{21}}\right)}\right) \leq c P^{-M}
\end{align*}}for some constant $c>0$, when the difference matrices $\triangle {X'^{k_1,k_2}_{i1}}$ are full rank for all $k_1 \neq k_2$ and for $i=1,2$.
\end{theorem}
\begin{proof}
 The proof is a generalization of the proof of Theorem $4$ in \cite{AbR}. The proof is given in Appendix \ref{appen_thm1}.
\end{proof}

An STBC $X'$ for $M=3$ which possesses the column cancellation property is given in (\ref{eqn-M=3_STBC}). 
{\begin{figure*}
\begin{align} \label{eqn-M=3_STBC}
   X'= \begin{bmatrix}
        x^{1R}+jx^{3I} & -x^{2R}+jx^{4I} &e^{j\theta}\left(x^{5R}+jx^{6I}\right) & e^{j\theta}\left(-x^{3R}+jx^{1I}\right) \\
	x^{2R}+jx^{4I} &  x^{1R}-jx^{3I} & e^{j\theta}\left(x^{4R}+jx^{2I}\right) & e^{j\theta}\left(x^{5R}-jx^{6I}\right) \\
	e^{j\theta}\left(x^{6R}+jx^{5I}\right) & e^{j\theta}\left(-x^{6R}+jx^{5I}\right)& x^{3R}+jx^{1I} & -x^{4R}+jx^{2I} \\
     \end{bmatrix}
\end{align}
\hrule
\end{figure*}}The matrices $X_{i1}$ and $X_{i2}$ in the $(2,2,3)-X$-Network are given by (\ref{eqn-Abhi_3tx_Code_for_X_Ch1}) and (\ref{eqn-Abhi_3tx_Code_for_X_Ch2}) respectively. It is assumed that the symbols $x^k_{ij}$ take values from a finite constellation ${\cal S}_{ij}$ whose average energy is equal to one. The constants $c_1$ and $c_2$ are given by $c_1=c_2=\frac{3}{4}$.
{\begin{figure*}
\begin{align} \label{eqn-Abhi_3tx_Code_for_X_Ch1} 
   X_{i1}= \begin{bmatrix}
        x^{1R}_{i1}+jx^{3I}_{i1} & -x^{2R}_{i1}+jx^{4I}_{i1} & 0 &e^{j\theta}\left(x^{5R}_{i1}+jx^{6I}_{i1}\right) & e^{j\theta}\left(-x^{3R}_{i1}+jx^{1I}_{i1}\right) & 0 \\
	x^{2R}_{i1}+jx^{4I}_{i1} &  x^{1R}_{i1}-jx^{3I}_{i1} & 0 & e^{j\theta}\left(x^{4R}_{i1}+jx^{2I}_{i1}\right) & e^{j\theta}\left(x^{5R}_{i1}-jx^{6I}_{i1}\right)& 0 \\
	e^{j\theta}\left(x^{6R}_{i1}+jx^{5I}_{i1}\right) & e^{j\theta}\left(-x^{6R}_{i1}+jx^{5I}_{i1}\right) & 0 & x^{3R}_{i1}+jx^{1I}_{i1} & -x^{4R}_{i1}+jx^{2I}_{i1}& 0 \\
	\end{bmatrix}
\end{align}
\hrule
\end{figure*}}
{\begin{figure*}
\begin{align} \label{eqn-Abhi_3tx_Code_for_X_Ch2} 
   X_{i2}= \begin{bmatrix}
        0 & x^{1R}_{i2}+jx^{3I}_{i2} & -x^{2R}_{i2}+jx^{4I}_{i2} & 0 &e^{j\theta}\left(x^{5R}_{i2}+jx^{6I}_{i2}\right) & e^{j\theta}\left(-x^{3R}_{i2}+jx^{1I}_{i2}\right) \\
	0 & x^{2R}_{i2}+jx^{4I}_{i2} &  x^{1R}_{i2}-jx^{3I}_{i2} & 0 & e^{j\theta}\left(x^{4R}_{i2}+jx^{2I}_{i2}\right) & e^{j\theta}\left(x^{5R}_{i2}-jx^{6I}_{i2}\right)\\
	0 & e^{j\theta}\left(x^{6R}_{i2}+jx^{5I}_{i2}\right) & e^{j\theta}\left(-x^{6R}_{i2}+jx^{5I}_{i2}\right) & 0 & x^{3R}_{i2}+jx^{1I}_{i2} & -x^{4R}_{i2}+jx^{2I}_{i2}\\
     \end{bmatrix}
\end{align}
\hrule
\end{figure*}}

The interference cancellation at Rx-$1$ is done as follows. The matrix $Y'_1 \in \mathbb{C}^{3 \times 4}$ obtained by processing $Y_1 \in \mathbb{C}^{3 \times 6}$ is given below.

{\small
\begin{align}
\nonumber
& Y'_1(:,1)=Y_1(:,1),\\
\nonumber
& Y'_1(:,3)=Y_1(:,4),\\
\nonumber
& Y'_1(1,2)=Y_1(1,2) - \overline{Y_1(2,3)},\\ 
\nonumber
& Y'_1(2,2)=Y_1(2,2) + \overline{Y_1(1,3)},\\
\label{eqn-second_col3}
& Y'_1(3,2)=Y_1(3,2) + e^{j2\theta}\overline{Y_1(3,3)}, \\
\nonumber
& Y'_1(1,4)=Y_1(1,5) - e^{j2\theta}\overline{Y_1(2,6)},\\ 
\nonumber
& Y'_1(2,4)=Y_1(2,5) + e^{j\theta}\overline{Y_1(3,6)},\\
\nonumber
& Y'_1(3,4)=Y_1(3,5)  + e^{j\theta}\overline{Y_1(1,6)}.		
\end{align}}Note that the LJJ scheme for $(2,2,2)-X$-Network also involves similar interference cancellation procedure though it was explained through zero-forcing of aligned interference in Section \ref{subsec2}.

It is observed that the proposed scheme encodes a total of $12$ complex symbols at every transmitter in $6$ time slots and hence, a sum throughput of $4$ cspcu is achieved in the $(2,2,2)-X$-Network.

From Theorem \ref{thm1}, to show that ML decoding of $X_{i1}$ from $Y'_1$ given by (\ref{eqn-after_int_cal}) yields a diversity gain of $3$, we need to prove that for any finite constellation input there always exists $\theta$ such that the difference matrix ${\triangle X'^{k_1,k_2}_{i1}}$ is of rank $3$, for all $k_1 \neq k_2$. Towards that end, we have the following definition from \cite{ZaR}.

\begin{definition}\cite{ZaR}
  The Coordinate Product Distance (CPD) between any two signal points $u=u^R+ju^I$ and $v=v^R+jv^I$, for $u \neq v$, in a finite constellation ${\cal S}$ is defined as 
\begin{align*}
 CPD(u,v)=\left|u^R-v^R\right|\left|u^I-v^I\right|
\end{align*}and the minimum of this value among all possible pairs is defined as the CPD of ${\cal S}$.
\end{definition}We assume that each symbol $x^{k}_{ij}$ takes values from a finite constellation whose CPD is non-zero, for all $i,j,k$. If a finite constellation has a zero CPD, it can be rotated appropriately so that the resulting constellation has a non-zero CPD \cite{ZaR}.

\begin{lemma} \label{lem1}
There always exists $\theta \in [0,2\pi]$ such that the difference matrix ${\triangle X'^{k_1,k_2}_{ij}}$ is full-rank (i.e., rank $=3$), for all $k_1 \neq k_2$ and $i,j=1,2$.
\end{lemma}
\begin{proof}
 Proof is given in Appendix \ref{appen_lem1}.
\end{proof}

From the proof of Lemma $1$, it can be observed that selecting $\theta$ randomly, for instance from the uniform distribution in $[0,2\pi]$, would ensure that the difference matrix ${\triangle X'^{k_1,k_2}_{ij}}$ is full-rank, for all $k_1 \neq k_2$ and $i,j=1,2$. Thus, $\theta$ can be chosen easily to yield a diversity gain of $3$.

We now prove that using Gaussian input alphabets and coding across time, it is possible to achieve a sum DoF of $4$ using the proposed scheme.

\begin{theorem} \label{thm2}
The proposed scheme can achieve a sum DoF of $4$ with symbol-by-symbol decoding, for any $\theta \in [0,2\pi]$.
\end{theorem}
\begin{proof}
Proof is given in Appendix \ref{appen_thm2}.
\end{proof}

\begin{remark}
 The STBC used in the proposed scheme is inspired by the SR STBC \cite{SrR1} (given in (\ref{eqn-SR_STBC})) which also possesses the column cancellation property and hence, used in $(2,2,4)-X$-Network \cite{AbR}. 
 {\begin{figure*}
\begin{align} \label{eqn-SR_STBC} 
   X'= \begin{bmatrix}
        x^{1R}+jx^{3I} & -x^{2R}+jx^{4I} & e^{j\theta}\left(x^{5R}+jx^{7I}\right) & e^{j\theta}\left(-x^{6R}+jx^{8I}\right) \\
	x^{2R}+jx^{4I} &  x^{1R}-jx^{3I} & e^{j\theta}\left(x^{6R}+jx^{8I}\right) & e^{j\theta}\left(x^{5R}-jx^{7I}\right) \\
	e^{j\theta}\left(x^{7R}+jx^{5I}\right) & e^{j\theta}\left(-x^{8R}+jx^{6I}\right) & x^{3R}+jx^{1I} & -x^{4R}+jx^{2I} \\
	e^{j\theta}\left(x^{8R}+jx^{6I}\right) & e^{j\theta}\left(x^{7R}-jx^{5I}\right) & x^{4R}+jx^{2I} &  x^{3R}-jx^{1I}\\
     \end{bmatrix}
\end{align}
\hrule
\end{figure*}}However, it can be observed from that (\ref{eqn-M=3_STBC}) cannot be trivially obtained from the SR STBC by deleting one of its rows because the two STBCs are meant to offer different throughputs (in cspcu). 
\end{remark}

Some simulation plots are shown in the next section, comparing the performance of the proposed scheme with the JS scheme. 

\section{Simulation Results} \label{sec5}
The bit error rate (BER) of the proposed scheme using QPSK\footnote{Gray labeling is used on all the constellations in this paper.} input constellations at all the transmitters in the $(2,2,3)-X$-Network is plotted in Fig. \ref{fig-QPSK}. We set $\theta=\frac{\pi}{4}$ in the proposed scheme, and the constellations are rotated by an angle $\phi=\frac{\text{tan}^{-1}(2)}{2}$ to ensure a non-zero CPD \cite{ZaR}. A brute force search in the software MATLAB was done to ensure that $\theta=\frac{\pi}{4}$ gives full-rank difference matrices ${\triangle X'^{k_1,k_2}_{ij}}$, for all $k_1 \neq k_2$ and $i,j=1,2$ so that Theorem \ref{thm1} is valid for this case. The transmitted symbols  are decoded using the sphere decoder \cite{ViB}. It is observed from Fig. \ref{fig-QPSK} that the proposed scheme achieves a diversity gain that is strictly greater than $3$. We refer to the proposed scheme with $\theta=0$ and constellation rotation angle $\phi=0$ as the Alamouti Repetition (AR) scheme. In the AR scheme, we note that 
the 
difference matrices ${\triangle X'^{k_1,k_2}_{ij}}$, for $i,j=1,2$, are 
not full-rank for some $k_1\neq k_2$ for any input constellation with independent real and imaginary parts (for e.g., QAM constellation). Thus, Theorem \ref{thm1} is not applicable in the case of AR scheme. In the JS scheme, MAP decoding of the desired symbols from (\ref{eqn-JS_Aligned}) reduces to ML decoding of all the symbols at high values of $P$ \cite{RaG}, i.e.,

{\small\begin{align*}
&(\hat{X}_{11},\hat{X}_{21}) = \arg \hspace{-0.7cm}\min_{X_{11},X_{21},X_{12}+X_{22}}\left|\left|Y'_1-\sqrt{\frac{3P}{2}}\left(H'_{11}V_{11}X_{11}\right.\right.\right.\\
&\hspace{3.2cm}\left.\left.\left.+H'_{21}V_{21}X_{21}\right)+H'_{11}V_{12}\left(X_{12}+X_{22}\right)\right|\right|^2\\
&(\hat{X}_{12},\hat{X}_{22}) = \arg \hspace{-0.7cm}\min_{X_{12},X_{22},X_{11}+X_{21}}\left|\left|Y'_2-\sqrt{\frac{3P}{2}}\left(H'_{12}V_{12}X_{12}\right.\right.\right.\\
&\hspace{3.2cm}\left.\left.\left.+H'_{22}V_{22}X_{22}\right)+H'_{12}V_{11}\left(X_{11}+X_{21}\right)\right|\right|^2.
\end{align*}}Hence, as noted in \cite{RaG} sphere decoder can be used when QAM constellations are employed. 

It is seen from Fig. \ref{fig-QPSK} that the proposed scheme with $\theta=\frac{\pi}{4}$ comfortably outperforms the AR scheme and the JS scheme in terms of BER.

 {\begin{figure}[htbp]
\centering
\includegraphics[totalheight=2.2in,width=3.7in]{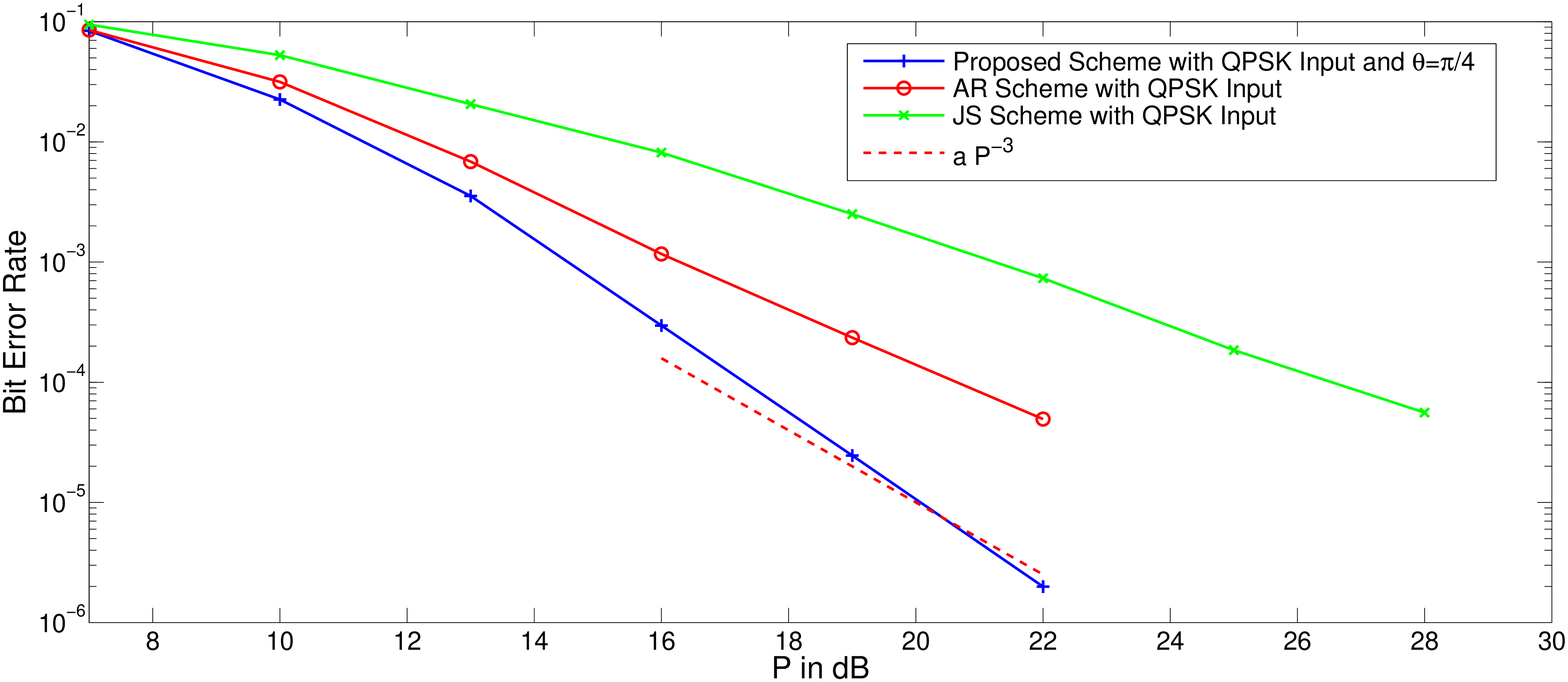}
\caption{Plot of BER vs P (in dB) for QPSK input constellations where the throughput is given by $4$ bits/sec/Hz per transmitter.  The dotted red line marked by $aP^{-3}$ is plotted for some positive real number $a$.}
\label{fig-QPSK}
\vspace{-0.5cm}
\end{figure}}
 {\begin{figure}[htbp]
\centering
\includegraphics[totalheight=2.2in,width=3.7in]{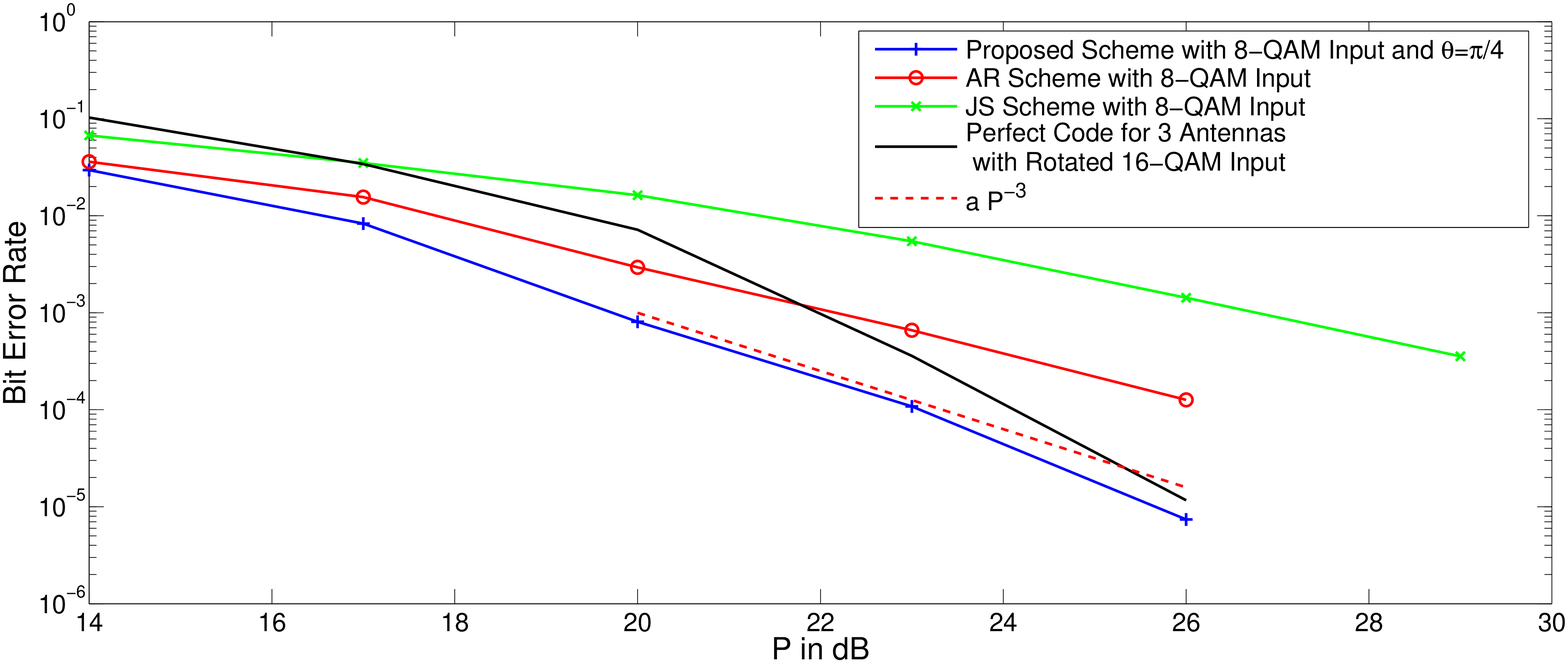}
\caption{Plot of BER vs P (in dB) where the proposed scheme with $\theta=\frac{\pi}{4}$, the AR scheme and the JS scheme use $8$-QAM input constellations, and time sharing with Perfect code for $3$ antennas uses $16$-QAM input constellations so that the throughput is given by $6$ bits/sec/Hz per transmitter. The dotted red line marked by $aP^{-3}$ is plotted for some positive real number $a$.}
\label{fig-8_QAM}
\end{figure}}

We now compare the BER performance of the proposed scheme with a time-sharing scheme in the $(2,2,3)-X$-Network. In order to equate the throughputs of the schemes, we use $8$-QAM constellation for the proposed scheme and $16$-QAM constellation for the time sharing scheme which is expected to transmit at a throughput of $3$ cspcu. Since we assumed CSIT in the proposed scheme, it is only fair to compare the proposed scheme with a time-sharing scheme using CSIT. Most of the research on single user MIMO systems with CSIT are focused on selecting a precoder that maximizes the minimum Euclidean distance at the receiver. One such work for the three antenna system (i.e. a system with $3$-Tx and $3$-Rx antennas) is \cite{NBS}. However, the work in \cite{NBS} is infeasible to implement in practice for the $16$-QAM case because of the huge number of difference matrices that need to be evaluated. Low complexity precoding techniques with CSIT for maximizing the minimum Euclidean distance at the receiver for single user 
systems 
in \cite{SrR2,MVHC} focus only on systems with even number of antennas.  So, we choose to compare the proposed scheme with a time-sharing scheme using Perfect STBC \cite{ORBV} which does not employ CSIT. Perfect STBC guarantees a diversity gain of $9$ in three antenna single user systems when the input constellations are chosen as a QAM constellation whose real and imaginary parts are post-multiplied by a lattice generator matrix given by $G=\begin{bmatrix}1 & 0\\
 \frac{1}{2} & \frac{\sqrt{3}}{2}
\end{bmatrix}$ \cite{ORBV}. The proposed scheme is simulated with $8$-QAM constellation rotated by an angle $\phi=\frac{\text{tan}^{-1}(2)}{2}$ to ensure a non-zero CPD, and $\theta=\frac{\pi}{4}$ for which it is verified that the difference matrices ${\triangle X'^{k_1,k_2}_{ij}}$, for all $k_1 \neq k_2$ and $i,j=1,2$, are full-rank. It is noted from Fig. \ref{fig-8_QAM} that the proposed scheme with $8$-QAM constellation and $\theta=\frac{\pi}{4}$ outperforms the Perfect STBC with $16$-QAM constellation rotated by the matrix $G$ in the chosen range of $P$. We also note that the AR scheme and the JS scheme using $8$-QAM constellations perform poorly compared to the proposed scheme. Once again Theorem \ref{thm1} is validated by Fig. \ref{fig-8_QAM}, where it is observed that the proposed scheme with $\theta=\frac{\pi}{4}$ achieves a diversity gain that is strictly greater than $3$.

\section{Conclusion} \label{sec6}
We extended the LJJ scheme to the $(2,2,3)-X$-Network using a newly proposed STBC for a three transmit antenna single user MIMO system. We showed that the proposed scheme achieves a diversity gain of at least $3$ with fixed finite constellation inputs and a sum DoF of $4$ with only local CSIT requirement. The JS scheme, on the other hand, required global CSIT to achieve the maximum sum DoF of $4$. 

If we could obtain STBCs with full-rank difference matrices, and with the proposed column cancellation property and a throughput of $\frac{M}{2}$ cspcu, the LJJ scheme could be extended to a general $(2,2,M)-X$-Network to achieve a diversity gain of $M$. However, this is not the only challenge. The main challenge is the decoding complexity for systems with $M>4$. The proposed scheme in this paper and the extended LJJ scheme using the SR STBC for the $(2,2,4)-X$-Network \cite{AbR} enjoyed the decoding of the transmitted symbols using the sphere decoder. However, with higher dimensions, even sphere decoding would be complicated because the choice of sphere radius becomes more critical in determining the feasibility of decoding in higher dimensional systems. Thus, extending the LJJ scheme to systems with $M>4$ should also consider decoding complexity as a criterion while designing STBCs with column cancellation property.

\appendices
\section{Proof of Theorem \ref{thm1}}
\label{appen_thm1}
\begin{IEEEproof}
Consider a modified system where a Gaussian noise matrix is added to (\ref{eqn-after_int_cal}) so that the entries of the effective noise matrix in (\ref{eqn-after_int_cal}) are distributed as i.i.d. ${\cal CN}(0,1+\max_{k}\sigma_k^2)$. Let $\sigma^2=1+\max_{k}\sigma_k^2$. The average pair-wise error probability for this modified system is given by

{\footnotesize\begin{align} 
\nonumber
& P_e\left({\left({X'^{k_1}_{11}},{X'^{k_2}_{21}}\right)\rightarrow\left({X'^{k'_1}_{11}},{X'^{k'_2}_{21}}\right)}\right)= \\
\label{eqn-PEP1}
& \mathbb{E}\left[Q\left(P\sqrt{\left|\left|\sqrt{c_1}H_{11}V_{11}\triangle X'_{11} + \sqrt{c_2}H_{21}V_{21}\triangle X'_{21}\right|\right|^2/2\sigma^2}\right)\right]
\end{align}}where, {\small$\triangle X'_{11}={X'^{k_1}_{11}}-{X'^{k'_1}_{11}}$}, and {\small$\triangle X'_{21}={X'^{k_2}_{21}}-{X'^{k'_2}_{21}}$}. Note that either {\small$\triangle X'_{11} \neq 0, \triangle X'_{21}=0$} or {\small$\triangle X'_{11} = 0, \triangle X'_{21}\neq 0$} or {\small$\triangle X'_{11} \neq 0, \triangle X'_{21}\neq 0$}. We shall prove the statement of the theorem only for the cases where {\small$\triangle X'_{11} \neq 0$} (i.e., the cases {\small$\triangle X'_{11} \neq 0, \triangle X'_{21}=0$} and {\small$\triangle X'_{11} \neq 0, \triangle X'_{21}\neq 0$}), and the proof for the case {\small$\triangle X'_{11} = 0, \triangle X'_{21}\neq 0$} is similar.

The Frobenius norm in (\ref{eqn-PEP1}) can be re-written as (\ref{eqn-H'}) (given at the top of the next page).
\begin{figure*} \begin{align}
\label{eqn-H'}
\left|\left|\sqrt{c_1}H_{11}V_{11}\triangle X'_{11} + \sqrt{c_2}H_{21}V_{21}\triangle X'_{21}\right|\right|^2 =&
 \left[ \underbrace{\left(\sqrt{c_1}{\triangle X'^T_{11}} V_{11}^T \otimes I_M \right)vec(H_{11})+\left(\sqrt{c_2}{\triangle X'^T_{21}} V_{21}^T \otimes I_M \right)vec(H_{21})}_{H'}\right]^H \times\\
 \nonumber
& \left[ \left(\sqrt{c_1}{\triangle X'^T_{11}} V_{11}^T \otimes I_M \right)vec(H_{11})+\left(\sqrt{c_2}{\triangle X'^T_{21}} V_{21}^T \otimes I_M \right)vec(H_{21})\right].
\end{align} \hrule
\end{figure*}Note that, conditioned on $H_{12}$ and $H_{22}$, the vector {\small$H'$} defined in (\ref{eqn-H'}) is a Gaussian vector with mean zero and covariance matrix $K=K'\otimes I_M$, where 

{\footnotesize \begin{align*}  
K'=c_1\left({\triangle X'^T_{11}} V_{11}^T\right)\left({\triangle X'^T_{11}} V_{11}^T\right)^H \hspace{-0.2cm}+c_2\left({\triangle X'^T_{21}} V_{21}^T\right)\left({\triangle X'^T_{21}} V_{21}^T\right)^H.
\end{align*}}{\begin{figure*}
 \footnotesize\begin{align}
\nonumber
&\mathbb{E}\left[Q\left(\sqrt{P\left|\left|\sqrt{c_1}H_{11}V_{11}\triangle X'_{11} + \sqrt{c_2}H_{21}V_{21}\triangle X'_{21}\right|\right|^2/2\sigma^2}\right)\right]\\
\label{eqn-div1}
&=\mathbb{E}_{H_{12},H_{22}}\left[\mathbb{E}_{H_{11},H_{21}|{H_{12},H_{22}}}\left[Q\left(\sqrt{P\left|\left|\sqrt{c_1}H_{11}V_{11}\triangle X'_{11} + \sqrt{c_2}H_{21}V_{21}\triangle X'_{21}\right|\right|^2/2\sigma^2}\right)\right]\right]\\
\label{eqn-div2}
&=\mathbb{E}_{H_{12},H_{22}}\left[\mathbb{E}_{H''|{H_{12},H_{22}}}\left[Q\left(\sqrt{P'\frac{H''^H H''}{2\sigma^2}}\right)\right]\right]=\mathbb{E}_{H_{12},H_{22}}\left[\mathbb{E}_{H_1,H_2,\cdots,H_M|{H_{12},H_{22}}}\left[Q\left(\sqrt{P'\frac{\sum_{i=1}^{M}\left|\left|K'^{\frac{1}{2}}H_i\right|\right|^2}{2\sigma^2}}\right)\right]\right]\\
\label{eqn-div4}
&=\mathbb{E}_{H_{12},H_{22}}\left[\mathbb{E}_{H_1,H_2,\cdots,H_M|{H_{12},H_{22}}}\left[Q\left(\sqrt{P'\frac{\sum_{i=1}^{M}\text{tr}\left(H_i^H {K'^{\frac{1}{2}}}^H K'^{\frac{1}{2}} H_i \right)}{2\sigma^2}}\right)\right]\right]\\
\label{eqn-div5}
&=\mathbb{E}_{H_{12},H_{22}}\left[\mathbb{E}_{H_1,H_2,\cdots,H_M|{H_{12},H_{22}}}\left[Q\left(\sqrt{P'\frac{\sum_{i=1}^{M}\text{tr}\left(H_i^H \Lambda H_i \right)}{2\sigma^2}}\right)\right]\right]\\
\label{eqn-div8}
&=\mathbb{E}_{H_{12},H_{22}}\left[\mathbb{E}_{H_1,H_2,\cdots,H_M|{H_{12},H_{22}}}\left[Q\left(\sqrt{P'\frac{\sum_{i=1}^{M}\sum_{j=1}^{T'} \lambda_j(K') |H_i(j)|^2 }{2\sigma^2}}\right)\right]\right]\\
\label{eqn-div9}
&\leq\mathbb{E}_{H_{12},H_{22}}\left[\mathbb{E}_{H_1,H_2,\cdots,H_M|{H_{12},H_{22}}}\left[Q\left(\sqrt{P'\frac{\sum_{i=1}^{M}\sum_{j=1}^{T'} \lambda_j(K'_1) |H_i(j)|^2 }{2\sigma^2}}\right)\right]\right]\\
\label{eqn-div10}
&=\mathbb{E}_{H_{12}}\left[\mathbb{E}_{H'_1,H'_2,\cdots,H'_M|{H_{12}}}\left[Q\left(\sqrt{P'\frac{\sum_{i=1}^{M}\text{tr}\left(H'^H_i \Lambda'_1 H'_i \right)}{2\sigma^2}}\right)\right]\right]\\
\label{eqn-div11}
&=\mathbb{E}_{H_{12}}\left[\mathbb{E}_{H'_1,H'_2,\cdots,H'_M|{H_{12}}}\left[Q\left(\sqrt{P'\frac{\sum_{i=1}^{M}\text{tr}\left(\left(V_1^HH'_i\right)^H {\Lambda'_1} \left(V_1^H H'_i\right) \right)}{2\sigma^2}}\right)\right]\right]\\
\label{eqn-div12}
&=\mathbb{E}_{H_{12}}\left[\mathbb{E}_{H'_1,H'_2,\cdots,H'_M|{H_{12}}}\left[Q\left(\sqrt{P'\frac{\sum_{i=1}^{M}\text{tr}\left(H'^H_i \left(V_1 {{\Lambda'_1}^{\frac{1}{2}}}^HU_1^H\right) ~\left(U_1 {\Lambda'_1}^{\frac{1}{2}} V_1^H\right) H'_i \right)}{2\sigma^2}}\right)\right]\right]\\
\label{eqn-div13}
&=\mathbb{E}_{H_{12}}\left[\mathbb{E}_{H'_1,H'_2,\cdots,H'_M|{H_{12}}}\left[Q\left(\sqrt{P'\frac{\sum_{i=1}^{M}\left|\left|{K'_1}^{\frac{1}{2}}H'_i\right|\right|^2}{2\sigma^2}}\right)\right]\right]\\
\label{eqn-div14}
&=\mathbb{E}_{H_{12}}\left[\mathbb{E}_{H'_1,H'_2,\cdots,H'_M|{H_{12}}}\left[Q\left(\sqrt{P'\frac{\sum_{i=1}^{M} H'^H_i{V^T_{11}}^H \left(\triangle X'_{11} {\triangle X'_{11}}^H\right)^TV^T_{11}H'_i}{2\sigma^2}}\right)\right]\right]\\
\label{eqn-div15}
&=\mathbb{E}_{H_{12}}\left[\mathbb{E}_{H'_1,H'_2,\cdots,H'_M|{H_{12}}}\left[Q\left(\sqrt{P'\frac{\sum_{i=1}^{M} H'^H_i\left(\left(V_{11}U_{\triangle X'_{11}}\right)^T\right)^H \Lambda'_{\triangle X'_{11}} \left(V_{11}U_{\triangle X'_{11}}\right)^T H'_i}{2\sigma^2}}\right)\right]\right]
\end{align}\hrule
\end{figure*}}{\begin{figure*}
 \footnotesize\begin{align}
 \label{eqn-div16}
&\leq \mathbb{E}_{H_{12}}\left[\mathbb{E}_{H'_1,H'_2,\cdots,H'_M|{H_{12}}}\left[Q\left(\sqrt{P' \lambda_M\left(\triangle X'_{11}\right)\frac{\sum_{i=1}^{M} H'^H_i\left(\left(V_{11}U_{\triangle X'_{11}}\right)^T\right)^H \left(V_{11}U_{\triangle X'_{11}}\right)^T H'_i}{2\sigma^2}}\right)\right]\right]\\
\label{eqn-div17}
&= \mathbb{E}_{H_{12}}\left[\mathbb{E}_{H'_1,H'_2,\cdots,H'_M|{H_{12}}}\left[Q\left(\sqrt{P' \lambda_M\left(\triangle X'_{11}\right)\frac{\sum_{i=1}^{M} H'^H_i\left(V^T_{11}\right)^H V^T_{11} H'_i}{2\sigma^2}}\right)\right]\right]\\
\label{eqn-div18}
&= \mathbb{E}_{H_{12}}\left[\mathbb{E}_{H'_1,H'_2,\cdots,H'_M|{H_{12}}}\left[Q\left(\sqrt{P' \lambda_M\left(\triangle X'_{11}\right)\frac{\sum_{i=1}^{M} H'^H_i U_{V_{11}} \Lambda_{V_{11}} U^H_{V_{11}} H'_i}{2\sigma^2}}\right)\right]\right]\\
\label{eqn-div19}
&= \mathbb{E}_{H_{12}}\left[\mathbb{E}_{H_1,H_2,\cdots,H'_M|{H_{12}}}\left[Q\left(\sqrt{P' \lambda_M\left(\triangle X'_{11}\right)\frac{\sum_{i=1}^{M} H'^H_i U_{V_{11}} \Lambda_{V_{11}} U^H_{V_{11}} H'_i}{2\sigma^2}}\right)\right]\right]\\
\label{eqn-div20}
&= \mathbb{E}_{H_{12}}\left[\mathbb{E}_{H_1,H_2,\cdots,H'_M|{H_{12}}}\left[Q\left(\sqrt{P' \lambda_M\left(\triangle X'_{11}\right)\frac{\sum_{i=1}^{M} \left(U_{V_{11}}^H H'_i\right)^H  \Lambda_{V_{11}} U^H_{V_{11}} H'_i}{2\sigma^2}}\right)\right]\right]\\
\label{eqn-div21}
&\stackrel{(a)}{\leq} \mathbb{E}_{H_{12}}\left[\frac{1}{\prod_{j=1}^M\left(1+\frac{c_1P\lambda_M(\triangle X'_{11})\lambda_j(V_{11})}{2\sigma^2}\right)^M}\right] \stackrel{(b)}{<} \frac{1}{\left(1+\frac{c_1P\lambda_M(\triangle X'_{11})}{2\sigma^2M}\right)^M} \stackrel{(c)}{\approx} cP^{-M}
\end{align}
\hrule
\end{figure*}}In other words, when the successive elements of {\small$H'$} are grouped in blocks of $T'$ entries each, the blocks are distributed i.i.d. as Gaussian matrix with zero mean and covariance matrix given by $K'$. Since $K'$ is a positive semi-definite Hermitian matrix, let the eigen decomposition of the matrix $K'$ be given by $K'=U \Lambda U^H$ where, $U$ is a $T' \times T'$ unitary matrix formed by the eigen vectors of $K'$, and {\small$\Lambda=\text{diag}\left(\lambda_1(K'),\lambda_2(K'),\cdots,\lambda_{T'}(K')\right)$} denotes the matrix whose diagonal entries are ordered eigen values of $K'$ with {\small$\lambda_1(K')\geq\lambda_2(K')\geq \cdots \geq\lambda_{T'}(K')\geq 0$}. Denote a square-root of {\small$K'$} by {\small$K'^{\frac{1}{2}}$}, i.e., {\small$K'=K'^{\frac{1}{2}}{K'^{\frac{1}{2}}}^H$} where, {\small$K'^{\frac{1}{2}}=U \Lambda^{\frac{1}{2}}$}. The vector $H'$ is now statistically equivalent to the following vector

{\small\begin{align*}
 H'' = \begin{bmatrix}
       K'^{\frac{1}{2}}H_1\\
        K'^{\frac{1}{2}}H_2\\
	 \vdots\\
	   K'^{\frac{1}{2}}H_M\\
      \end{bmatrix}
\end{align*}}where, $H_i \in \mathbb{C}^{T' \times 1}$, $i=1,2,\cdots,M$, are Gaussian vectors whose entries are distributed as i.i.d. ${\cal CN}(0,1)$. Now, (\ref{eqn-PEP1}) can be successively re-written as in (\ref{eqn-div1})-(\ref{eqn-div10}) where, (\ref{eqn-div2}) follows from the statistical equivalence between $H'$ and $H''$, (\ref{eqn-div4}) follows from the fact that $||A||^2= \text{tr}(A^HA)$, and (\ref{eqn-div5}) follows from the definition of $K'^{\frac{1}{2}}$. Now, define {\small$K'_1=c_1\left({\triangle X'^T_{11}} V_{11}^T\right)\left({\triangle X'^T_{11}} V_{11}^T\right)^H$} and {\small$K'_2=c_2\left({\triangle X'^T_{21}} V_{21}^T\right)\left({\triangle X'^T_{21}} V_{21}^T\right)^H$} so that {\small$K'=K'_1+K'_2$}. Let $\lambda_j(K'_1)$ denote the eigen values of $K'_1$ in non-increasing order from $j$$=$$1$ to $j$$=$$T'$. Using Weyl's inequalities \footnote{Weyl's inequalities relate the eigen values of sum of two Hermitian matrices with the eigen values of the individual matrices.} (see 
Section III.2, pp. $62$ of \cite{Bha}), we have {\small$\lambda_j(K'_1) \leq \lambda_j(K')$}, $j=1,2,\cdots,T'$. Thus, we have the inequality (\ref{eqn-div9}) from (\ref{eqn-div8}) where, $H_i(j)$ denotes the $j^{\text{th}}$ entry of the vector $H_i$. Let $K'_1=U_1 \Lambda_1 U_1^H$ denote the eigen decomposition of $K'_1$, where\footnote{The diagonal elements of $\Lambda_1$ are ordered in a non-increasing order.} {\small$\Lambda_1=\text{diag}(\lambda_1(K'_1),\lambda_2(K'_1),\cdots,\lambda_M(K'_1),0,\cdots,0) \in \mathbb C^{T' \times T'}$}, and $U_1$ is a unitary matrix composed of eigen vectors of $K'_1$. The last $(T'-M)$ eigen values of $\Lambda_1$ are zero because $T'\geq M$ and the matrix $\triangle X'_{11}$ is of size $M \times T'$. Let $H'_i$ represent the first $M$ components of the $T'$-length vector $H_i$ and let {\small$\Lambda'_1=\text{diag}(\lambda_1(K'_1),\lambda_2(K'_1),\cdots,\lambda_M(K'_1))$}. Equation (\ref{eqn-div10}) follows from the fact that the argument inside the Q-function in (\ref{eqn-div9}) is independent of $H_{22}$ and the fact that the last $(T'-M)$ eigen values of $K'_1$ are zero. Let the singular value decomposition of $\triangle X'^T_{11} V^T_{11}$ be given by $\triangle X'^T_{11} V^T_{11}=U_1 {\Lambda'_1}^{\frac{1}{2}} V^H_1$, where
\begin{align*}
 {\Lambda'_1}^{\frac{1}{2}}=\begin{bmatrix}
                            \text{diag}\left(\sqrt{\lambda_1(K'_1)},\sqrt{\lambda_2(K'_1)},\cdots,\sqrt{\lambda_M(K'_1)}\right)\\
                            \mathbf{0}_{(T'-M)\times M}
                          \end{bmatrix},
\end{align*} $U_1\in \mathbb C^{T' \times T'}$ and $V_1\in \mathbb C^{M \times M}$ are unitary matrices. We observe that ${{\Lambda'_1}^{\frac{1}{2}}}^H{\Lambda'_1}^{\frac{1}{2}}=\Lambda'_1$. Note that $\triangle X'^T_{11} V^T_{11}$ is a square-root of $K'_1$ and hence, we shall denote this by ${K'_1}^{\frac{1}{2}}$. Now, (\ref{eqn-div11}) follows from the fact that the distribution of $H'_i$ is invariant to multiplication by the unitary matrix $V^H_1$, and using straight-forward simplifications we obtain (\ref{eqn-div14}). Now, let the eigen decomposition of {\small$\triangle X'_{11}\triangle X'^H_{11}$} be given by {\small$\triangle X'_{11}\triangle X'^H_{11}=U_{\triangle X'_{11}} \Lambda_{\triangle X'_{11}} U^H_{\triangle X'_{11}}$} where, {\small$\Lambda_{\triangle X'_{11}}$} denotes the eigen value matrix whose eigen values in non-increasing order are given by $\lambda_j\left(\triangle X'_{11}\right)$, $j=1,2,\cdots,M$. Note that $\lambda_M\left(\triangle X'_{11}\right)>0$ as it is assumed that $\triangle X'_{11}$ is full rank. Now, substitution of this eigen decomposition in (\ref{eqn-div14}) gives (\ref{eqn-div15}). The inequality (\ref{eqn-div16}) follows from the fact that $\lambda_M\left(\triangle X'_{11}\right)$ is the minimum eigen value of $\triangle X'_{11}$, and (\ref{eqn-div17}) follows from {\small$V_{11}$} being equal to {\small$\frac{H^{-1}_{12}}{\sqrt{\text{tr}\left(H^{-1}_{12}H^{-H}_{12}\right)}}$} and the fact that 
the distribution of $V_{11}$ is invariant to multiplication by the unitary matrix $U_{\triangle X'_{11}}$ (because $H_{12}$ is Gaussian distributed). Using the eigen decomposition of $\left(V^T_{11}\right)^H V_{11}=U_{V_{11}} \Lambda_{V_{11}} U_{V_{11}}$ and some straight-forward techniques involved in evaluating diversity as in \cite{TSC}, we obtain (\ref{eqn-div21})$(a)$. Now, note that the eigen values of $V_{11}$ are given by 
\begin{align*}
\lambda_j\left(V_{11}\right)=\frac{\frac{1}{\lambda_{M+1-j}\left(H_{12}\right)}}{\sum_{j=1}^{M}\frac{1}{\lambda_j\left(H_{12}\right)}} 
\end{align*}where, $\lambda_j\left(H_{12}\right)$ denote the eigen values of $H_{12}H^H_{12}$ in non-increasing order from $j=1$ to $j=M$. Thus, $\lambda_j\left(V_{11}\right)$ can be lower bounded as
\begin{align*}
 &\lambda_j\left(V_{11}\right) \geq \frac{\frac{1}{\lambda_{M+1-j}\left(H_{12}\right)}}{\sum_{j=1}^{M}\frac{1}{\lambda_M\left(H_{12}\right)}} = \frac{\lambda_M\left(H_{12}\right)}{M\lambda_{M+1-j}\left(H_{12}\right)}.
\end{align*}For $j=1$, the above lowerbound is equal to $\frac{1}{M}$, and for $j=2,3,\cdots, M$ the above lowerbound is in turn trivially lowerbounded by $0$. Hence, we obtain the inequality in (\ref{eqn-div21})$(b)$, and the approximation in (\ref{eqn-div21})$(c)$ holds good at high values of $P$, where the constant $c=\frac{(2\sigma^2M)^M}{c_1^M\lambda^M_M\left(\triangle X'_{11}\right)}$.
\end{IEEEproof}

\section{Proof of Lemma \ref{lem1}}
\label{appen_lem1}
\begin{IEEEproof}
We prove that for every difference matrix there exists atmost a finite number of values of $\theta$ for which it is not full-rank. Since there are infinite possible values of $\theta$, there always exists $\theta$ such that all the difference matrices are full-rank. 

Without loss of generality, we consider the difference matrix $\triangle {X'^{k_1,k_2}_{11}}$ for some $k_1\neq k_2$. Let the entries of the difference matrix be given by (\ref{eqn-diff_mat}) (at the top of the next page).
\begin{figure*}
\begin{align}\label{eqn-diff_mat}
   \triangle {X'^{k_1,k_2}_{11}}= \begin{bmatrix}
        \triangle x^{1R}_{11}+j\triangle x^{3I}_{11} & -\triangle x^{2R}_{11}+j\triangle x^{4I}_{11} &e^{j\theta}\left(\triangle x^{5R}_{11}+j\triangle x^{6I}_{11}\right) & e^{j\theta}\left(-\triangle x^{3R}_{11}+j\triangle x^{1I}_{11}\right) \\
	\triangle x^{2R}_{11}+j\triangle x^{4I}_{11} &  \triangle x^{1R}_{11}-j\triangle x^{3I}_{11} & e^{j\theta}\left(\triangle x^{4R}_{11}+j\triangle x^{2I}_{11}\right) & e^{j\theta}\left(\triangle x^{5R}_{11}-j\triangle x^{6I}_{11}\right) \\
	e^{j\theta}\left(\triangle x^{6R}_{11}+j\triangle x^{5I}_{11}\right) & e^{j\theta}\left(-\triangle x^{6R}_{11}+j\triangle x^{5I}_{11}\right) & \triangle x^{3R}_{11}+j\triangle x^{1I}_{11} & -\triangle x^{4R}_{11}+j\triangle x^{2I}_{11}\\
	\end{bmatrix}
\end{align}
\hrule
\end{figure*}Consider the matrices $A,B \in \mathbb C^{3\times 3}$ comprised of the first three columns and the last three columns of $\triangle {X'^{k_1,k_2}_{11}}$ respectively. Expanding along the last column, the determinant of the matrix $A$ is given by (\ref{eqn-det_A}).
\begin{figure*} \small
\begin{align}
\nonumber
|A|=&e^{2j\theta}\left(\triangle x^{5R}_{11}+j\triangle x^{6I}_{11}\right)\left(\left(\triangle x^{2R}_{11}+j\triangle x^{4I}_{11} \right)\left(-\triangle x^{6R}_{11}+j\triangle x^{5I}_{11}\right) -\left(\triangle x^{1R}_{11}-j\triangle x^{3I}_{11}\right) \left(\triangle x^{6R}_{11}+j\triangle x^{5I}_{11}\right)\right)\\
\label{eqn-det_A}
&-e^{2j\theta}\left(\triangle x^{4R}_{11}+j\triangle x^{2I}_{11}\right)\left(\left(\triangle x^{1R}_{11}+j\triangle x^{3I}_{11} \right)\left(-\triangle x^{6R}_{11}+j\triangle x^{5I}_{11}\right) -\left(-\triangle x^{2R}_{11}+j\triangle x^{4I}_{11}\right) \left(\triangle x^{6R}_{11}+j\triangle x^{5I}_{11}\right)\right)\\
\nonumber
&+\left(\triangle x^{3R}_{11}+j\triangle x^{1I}_{11}\right)\left(\triangle {x^{1R}_{11}}^2+{\triangle x^{3I}_{11}}^2+{\triangle x^{2R}_{11}}^2+{\triangle x^{4I}_{11}}^2\right)\\
\nonumber
|B|=&e^{j\theta}\left(-\triangle x^{2R}_{11}+j\triangle x^{4I}_{11}\right)\left(\left(\triangle x^{4R}_{11}+j\triangle x^{2I}_{11} \right)\left(-\triangle x^{4R}_{11}+j\triangle x^{2I}_{11}\right) -\left(\triangle x^{5R}_{11}-j\triangle x^{6I}_{11}\right) \left(\triangle x^{3R}_{11}+j\triangle x^{1I}_{11}\right)\right)\\
\label{eqn-det_B}
&-e^{j\theta}\left(\triangle x^{1R}_{11}-j\triangle x^{3I}_{11}\right)\left(\left(\triangle x^{5R}_{11}+j\triangle x^{6I}_{11} \right)\left(-\triangle x^{4R}_{11}+j\triangle x^{2I}_{11}\right) -\left(-\triangle x^{3R}_{11}+j\triangle x^{1I}_{11}\right) \left(\triangle x^{3R}_{11}+j\triangle x^{1I}_{11}\right)\right)\\
\nonumber
&+e^{2j\theta}\left(-\triangle x^{6R}_{11}+j\triangle x^{5I}_{11}\right)\left(\left(\triangle {x^{5R}_{11}}+j\triangle {x^{6I}_{11}}\right)\left(\triangle {x^{5R}_{11}}-j\triangle {x^{6I}_{11}}\right)-\left(-\triangle {x^{3R}_{11}}+j\triangle {x^{1I}_{11}}\right)\left(\triangle {x^{4R}_{11}}+j\triangle {x^{2I}_{11}}\right)\right)
\end{align}
\hrule
\end{figure*} Expanding along the first column, the determinant of the matrix $B$ is given by (\ref{eqn-det_B}). Since it is assumed that the CPD of the constellation involved is non-zero, $\triangle x^{iR}_{11}$ and $\triangle x^{iI}_{11}$ are either both zero or both non-zero, for some $i$. Now, consider the following cases.

\textbf{Case 1:} {\small$\left(\triangle x^{1R}_{11},\triangle x^{3R}_{11}\right)$$=$$(0,0)$} and {\small$\left(\triangle x^{5R}_{11},\triangle x^{6R}_{11}\right)$$=$$(0,0)$}. Here, the determinant of the matrix $B$ is given by
\begin{align*}
|B|=e^{j\theta}\left(-\triangle x^{2R}_{11}+j\triangle x^{4I}_{11}\right)\left(-{\triangle x^{4R}_{11}}^2-{\triangle x^{2I}_{11}}^2\right). 
\end{align*} Since $k_1\neq k_2$, either $\triangle x^{2R}_{11}$ or $\triangle x^{4I}_{11}$ or both of them are non-zero. Hence, $|B|\neq 0$ and $\triangle {X'^{k_1,k_2}_{11}}$ is of rank $3$.

\textbf{Case 2:} {\small$\left(\triangle x^{1R}_{11},\triangle x^{3R}_{11}\right)$$\neq$$(0,0)$} and {\small$\left(\triangle x^{5R}_{11},\triangle x^{6R}_{11}\right)$$=$$(0,0)$}. The determinant of the matrix $A$ is given by

{\small
\begin{align*}
 |A|=\left(\triangle x^{3R}_{11}+j\triangle x^{1I}_{11}\right)\left(\triangle {x^{1R}_{11}}^2+{\triangle x^{3I}_{11}}^2+{\triangle x^{2R}_{11}}^2+{\triangle x^{4I}_{11}}^2\right).
\end{align*}}Since $\triangle x^{3R}_{11}$ or $\triangle x^{1I}_{11}$ or both are non-zero, $|A| \neq 0$ for this case. Hence,  $\triangle {X'^{k_1,k_2}_{11}}$ is of rank $3$.

\textbf{Case 3:} {\small$\left(\triangle x^{1R}_{11},\triangle x^{3R}_{11}\right)$$=$$(0,0)$} and {\small$\left(\triangle x^{5R}_{11},\triangle x^{6R}_{11}\right)$$\neq$$(0,0)$}. In this case, the coefficient of $e^{2j\theta}$ in the determinant of the matrix $B$ is given by {\small$\left(-\triangle x^{6R}_{11}+j\triangle x^{5I}_{11}\right)\left({\triangle x^{5R}_{11}}^2+{\triangle x^{6I}_{11}}^2\right) \neq 0$}. Now, $|B|$ is a quadratic polynomial in $e^{j\theta}$ which can have atmost two roots for $e^{j\theta}$ and hence, atmost a finite number of values of $\theta$ for which $|B|=0$. Therefore, there exists infinite values of $\theta$ for which $|B| \neq 0$ in this case.

\textbf{Case 4:} {\small$\left(\triangle x^{1R}_{11},\triangle x^{3R}_{11}\right)$$\neq$$(0,0)$} and {\small$\left(\triangle x^{5R}_{11},\triangle x^{6R}_{11}\right)$$\neq$$(0,0)$}. If the first two terms of $|A|$ given in (\ref{eqn-det_A}) do not sum to zero then, $|A|$ is clearly a quadratic polynomial in $e^{j\theta}$. Thus, there exist infinite values of $\theta$ for which $|A|$ is non-zero. If the first two terms of $|A|$ sum to zero then,  $|A| \neq 0$ for the same reason as in Case $2$. Hence, $\triangle {X'^{k_1,k_2}_{11}}$ is of rank $3$ in this case also.
\end{IEEEproof}

\section{Proof of Theorem \ref{thm2}}
\label{appen_thm2}
\begin{IEEEproof}
 Consider a modified system where a Gaussian noise matrix is added to $Y'_1$ so that the entries of the effective noise matrix in (\ref{eqn-after_int_cal}) are distributed as i.i.d. ${\cal CN}(0,2)$. Define the matrices $H$ and $G$ by $H=H_{11}V_{11}$ and $G=H_{21}V_{21}$. The entries of the matrices $H$ and $G$ are denoted by $h_{ij}$ and $g_{ij}$ respectively, for $i,j=1,2,3$. Let the symbols $x^{k}_{i1}$ take values from the distribution ${\cal CN}(0,1)$. A processed and vectorized version of the first two columns of $Y'_1$ is given by (\ref{eqn-vec_y'_1})
 \begin{figure*}
 \begin{align}\label{eqn-vec_y'_1}
 \begin{bmatrix}
y'_{1_{11}}\\
\overline{y'_{1_{12}}}\\
y'_{1_{21}}\\
\overline{y'_{1_{22}}}\\
y'_{1_{31}}\\
\overline{y'_{1_{32}}}
\end{bmatrix}=\underbrace{
 \begin{bmatrix}
 h_{11} & h_{12} & e^{j\theta}h_{13} & g_{11} & g_{12} & e^{j\theta}g_{13}\\
 \overline{h_{12}} & -\overline{h_{11}} & -e^{-j\theta}\overline{h_{13}} & \overline{g_{12}} & -\overline{g_{11}} & -e^{-j\theta}\overline{g_{13}}\\
 h_{21} & h_{22} & e^{j\theta}h_{23} & g_{21} & g_{22} & e^{j\theta}g_{23}\\
 \overline{h_{22}} & -\overline{h_{21}} & -e^{-j\theta}\overline{h_{23}} & \overline{g_{22}} & -\overline{g_{21}} & -e^{-j\theta}\overline{g_{23}}\\
 h_{31} & h_{32} & e^{j\theta}h_{33} & g_{31} & g_{32} & e^{j\theta}g_{33}\\
 \overline{h_{32}} & -\overline{h_{31}} & -e^{-j\theta}\overline{h_{33}} & \overline{g_{32}} & -\overline{g_{31}} & -e^{-j\theta}\overline{g_{33}}
 \end{bmatrix}}_{R}\begin{bmatrix}
 p^1_{11}\\
 p^2_{11}\\
 p^3_{11}\\
 p^1_{21}\\
 p^2_{21}\\
 p^3_{21}
 \end{bmatrix} + N''_1
 \end{align} \hrule
 \end{figure*}where, $p^1_{i1}=x^{1R}_{i1}+jx^{3I}_{i1}$, $p^2_{i1}=x^{2R}_{i1}+jx^{4I}_{i1}$, $p^3_{i1}=x^{6R}_{i1}+jx^{5I}_{i1}$, and $y'_{1_{ij}}$ denotes the $i^{\text{th}}$ row, $j^{\text{th}}$ column element of $Y'_1$. Define the sub-matrices of the effective transfer matrix $R$ defined in (\ref{eqn-vec_y'_1}) by
 
{\footnotesize\begin{align*}
 &A_1= \begin{bmatrix}
h_{11} & h_{12} & e^{j\theta}h_{13}\\
\overline{h_{12}} & -\overline{h_{11}} & -e^{-j\theta}\overline{h_{13}}\\
h_{21} & h_{22} & e^{j\theta}h_{23} 
 \end{bmatrix}, B_1=\begin{bmatrix}
 g_{11} & g_{12} & e^{j\theta}g_{13}\\
 \overline{g_{12}} & -\overline{g_{11}} & -e^{-j\theta}\overline{g_{13}}\\
 g_{21} & g_{22} & e^{j\theta}g_{23}
 \end{bmatrix}\\
 &C_1= \begin{bmatrix} \overline{h_{22}} & -\overline{h_{21}} & -e^{-j\theta}\overline{h_{23}}\\
 h_{31} & h_{32} & e^{j\theta}h_{33}\\
 \overline{h_{32}} & -\overline{h_{31}} & -e^{-j\theta}\overline{h_{33}}
 \end{bmatrix},D_1= \begin{bmatrix}
 \overline{g_{22}} & -\overline{g_{21}} & -e^{-j\theta}\overline{g_{23}}\\
 g_{31} & g_{32} & e^{j\theta}g_{33}\\
 \overline{g_{32}} & -\overline{g_{31}} & -e^{-j\theta}\overline{g_{33}}
 \end{bmatrix}.
 \end{align*}}If it is shown that the matrix $R$ is almost surely full-rank for any value of $\theta$ then, the symbols $p^k_{i1}$, for $k=1,2,3$, can be decoded symbol-by-symbol, by zero-forcing the rest of the symbols almost surely. If it is proven that the determinant of $R$ is a non-zero polynomial in $h_{ij}$ and $g_{ij}$, $i,j=1,2,3$, for any value of $\theta$ then, the determinant is non-zero almost surely. This is because $h_{ij}$ and $g_{ij}$ are non-zero rational polynomial functions in the entries of the matrices $H_{ij}$, for $i,j=1,2$, which are continuously distributed. We now prove this by showing that, for any value of $\theta$, there exists an assignment of values to $h_{ij}$ and $g_{ij}$ such that the determinant of $R$ is non-zero\footnote{If the determinant of $R$ is a zero polynomial in $h_{ij}$ and $g_{ij}$, $i,j=1,2,3$, for some value of $\theta$ then, for any assignment to $h_{ij}$ and $g_{ij}$ the determinant would be equal to zero for that value of $\theta$.}.  Consider the 
following assignment of values to $h_{ij}$ and $g_{ij}$.
 \begin{align*}
 H=\begin{bmatrix}
    1 & 0 & 0\\
    0 & 1 & 1\\
    1 & 0 & 1
   \end{bmatrix},~G=\begin{bmatrix}
    0 & 0 & 0\\
    1 & 0 & 0\\
    1 & -e^{2j\theta} & 1
   \end{bmatrix}.
 \end{align*}The determinant of the matrix $R$ can be evaluated to be
 \begin{align*}
  |R|=|A||D-CA^{-1}B|=-2 \neq 0.
 \end{align*}Thus, the symbols $p^k_{i1}$, for $k=1,2,3$, can be decoded symbol-by-symbol almost surely, for any value of $\theta$.
 
 We still need to decode the symbols $p^k_{i1}$, $k=4,5,6$, given by $p^4_{i1}=x^{5R}_{i1}+jx^{6I}_{i1}$, $p^5_{i1}=x^{3R}_{i1}+jx^{1I}_{i1}$, and $p^6_{i1}=x^{4R}_{i1}+jx^{2I}_{i1}$. Consider a processed and vectorized version of the last two columns of $Y'_1$, given by (\ref{eqn-vec_y'_1_last_sym}).
 \begin{figure*}
 \begin{align}\label{eqn-vec_y'_1_last_sym}
 e^{-j\theta}
 \begin{bmatrix}
y'_{1_{13}}\\
\overline{y'_{1_{14}}}\\
y'_{1_{23}}\\
\overline{y'_{1_{24}}}\\
y'_{1_{33}}\\
\overline{y'_{1_{34}}}
\end{bmatrix}=\underbrace{
 \begin{bmatrix}
 h_{11} & e^{-j\theta}h_{13} & h_{12} & g_{11} & e^{-j\theta}g_{13} & g_{12}\\
 \overline{h_{12}} & -\overline{h_{11}} & -e^{-j\theta}\overline{h_{13}} & \overline{g_{12}} & -\overline{g_{11}} & -e^{-j\theta}\overline{g_{13}}\\
 h_{21} & e^{-j\theta}h_{23} & h_{22} & g_{21} & e^{-j\theta}g_{23} & g_{22}\\
 \overline{h_{22}} & -\overline{h_{21}} & -e^{-j\theta}\overline{h_{23}} & \overline{g_{22}} & -\overline{g_{21}} & -e^{-j\theta}\overline{g_{23}}\\
 h_{31} & e^{-j\theta}h_{33} & h_{32} & g_{31} & e^{-j\theta}g_{33} & g_{32}\\
 \overline{h_{32}} & -\overline{h_{31}} & -e^{-j\theta}\overline{h_{33}} & \overline{g_{32}} & -\overline{g_{31}} & -e^{-j\theta}\overline{g_{33}}\\
 \end{bmatrix}}_{S}\begin{bmatrix}
 p^4_{11}\\
 p^5_{11}\\
 p^6_{11}\\
 p^4_{21}\\
 p^5_{21}\\
 p^6_{21}
 \end{bmatrix} + N'''_1
 \end{align} \hrule
 \end{figure*}To prove that the symbols $p^k_{i1}$, for $k=4,5,6$, can be decoded symbol-by-symbol almost surely for any value of $\theta$, we need to show that there exists an assignment to $h_{ij}$ and $g_{ij}$ such that the determinant of the matrix $S$ defined in (\ref{eqn-vec_y'_1_last_sym}) is non-zero for any given value of $\theta$. So, consider the following assignment to $h_{ij}$ and $g_{ij}$.
  \begin{align*}
 H=I_3,~G=\begin{bmatrix}
    0 & 0 & 0\\
    1 & 0 & 0\\
    0 & 1 & 3-e^{j\theta}
   \end{bmatrix}.
 \end{align*}It can be verified that the determinant of $S$ is equal to $|S|=3\left(3-e^{j\theta}\right)\neq 0$, for any value of $\theta$. 
\end{IEEEproof}

\end{document}